\documentclass[10pt]{bmc_article}

\usepackage{cite} 
\usepackage{url}  
\usepackage{ifthen}  
\usepackage{multicol}   
\usepackage[utf8]{inputenc} 
\usepackage{color}
\usepackage{ulem}

\renewcommand{\emph}[1]{{\it {#1}}}

\newcommand{\figcite}[1]{Figure {#1}}

\newcommand{\figOSM}{1}
\newcommand{\figCayley}{2}
\newcommand{\figEgKlein}{3}
\newcommand{\figEgCyclic}{4}
\newcommand{\figAminoAcid}{5}

\usepackage{mathpazo} 
\usepackage{latexsym}
\usepackage{amssymb}
\usepackage{amsthm}
\usepackage{amsmath}
\usepackage{amsfonts,amsgen,amsopn}
\usepackage{color}
\usepackage{bm}
\usepackage[pdftex]{graphicx}
\graphicspath{{figures/}}

\newcommand{\T}{{\mathcal T}}

\newcommand{\C}{{\mathcal C}}
\newcommand{\NN}{{\mathbb N}}
\newcommand{\Z}{{\mathbb Z}}
\newcommand{\real}{\mathbb{R}}
\newcommand{\mS}{\Sigma}
\newcommand{\f}{\bm{f}}
\newcommand{\g}{\bm{g}}

\newtheorem{theorem}{Theorem}[section]

\newtheorem{example}[theorem]{Example}

\newtheorem{algorithm}{Algorithm}

\newboolean{publ}

\newenvironment{bmcformat}{\baselineskip20pt\sloppy\setboolean{publ}{false}}{\baselineskip20pt\sloppy}

\begin{document}
\begin{bmcformat}


\title{On the group theoretical background of assigning stepwise mutations onto phylogenies}


\author{Mareike Fischer$^1$%
         \email{email@mareikefischer.de}, 
         Steffen Klaere$^{2,}$\correspondingauthor%
         \email{steffen.klaere@gmail.com}, 
         Minh Anh Thi Nguyen$^1$
         \email{minh.anh.nguyen@univie.ac.at} and
         Arndt von Haeseler$^1$
         \email{arndt.von.haeseler@univie.ac.at} \\
         (authors are listed in alphabetical order)
      }


\address{%
    \iid(1)Center for Integrative Bioinformatics Vienna, Max F. Perutz Laboratories, University of Vienna, Medical University of Vienna, University of Veterinary Medicine Vienna, Dr. Bohr Gasse 9, A-1030, Vienna, Austria\\
    \iid(2)Department of Mathematics and Statistics, University of Otago, P.O. Box 56, Dunedin 9054, New Zealand
}%


\maketitle


\begin{abstract}
	In a recent paper \cite{klaere2008}, Klaere et al. modeled the impact of  substitutions on arbitrary branches of a phylogenetic tree on an alignment site by the so-called \textit{One Step Mutation} (OSM) matrix. 
	By utilizing the concept of the OSM matrix for the four-state nucleotide alphabet, Nguyen et al. \cite{nguyen2011} presented an efficient procedure to compute the minimal number of  substitutions needed to translate one alignment site into another.
	The present paper delivers a proof for this computation.
	Moreover, we provide several mathematical insights into the generalization of the OSM matrix to multistate alphabets.
	The construction of the OSM matrix is only possible if the matrices representing the substitution types acting on the character states and the identity matrix form a commutative group with respect to matrix multiplication.
	We illustrate a means to establish such a group for the twenty-state amino acid alphabet and critically discuss its biological usefulness.
\end{abstract}

\ifthenelse{\boolean{publ}}{\begin{multicols}{2}}{}


\section{Background}
\label{background}

Alignments of homologous sequences provide fundamental materials to the reconstruction of phylogenetic trees and many other sequence-based analyses (see, e.g., \cite{durbin1998,mount2004}). Each alignment column (site) consists of character states that are assumed to have evolved from a common ancestral state by means of mutations.
Any combination of the character states in the aligned sequences at one alignment column represents a so-called \textit{character} \cite{semple2003}, which is sometimes also called \textit{site pattern} \cite{nguyen2011}. Given a phylogenetic tree and an alignment that evolved along the tree, Klaere et al. \cite{klaere2008} showed, for binary alphabets, how a character changes into another character if a substitution occurs on an arbitrary branch of the tree. 
The impact of such a substitution is summarized by the so-called \textit{One Step Mutation} (OSM) matrix.
Nguyen et al. \cite{nguyen2011} extended the concept of the OSM matrix to the four-state nucleotide alphabet while developing a method to evaluate the goodness of fit between models and data in phylogenetic inference. 
There, the OSM matrix is constructed based on the Kimura three parameter (K3ST) substitution model \cite{kimura1981}.
Nguyen et al. \cite{nguyen2011} illustrated how one can use Maximum Parsimony (i.e. apply the Fitch algorithm \cite{fitch1971}) to compute the minimal number of  substitutions required to change one character into another character under the OSM setting. In the present paper, we deliver a proof for this transformation.

In addition, the OSM matrix can be constructed only if the substitution matrices and the identity matrix form a commutative or Abelian group (see, e.g., \cite{humphreys1996, keilen2001}) with respect to matrix multiplication \cite{nguyen2011}. 
We generalize the construction of the OSM matrix for any alphabet.
Moreover, we show that the number of  substitutions needed to convert one character into another may change if we use different groups. 
Finally, we provide a means to find an Abelian group for the twenty-state amino acid alphabet.
\section{Notation and Problem Recapitulation}
\label{notation-model}
\bigskip

\subsection{Notation}
\label{notation}
Recall that a {\it rooted binary phylogenetic $X$-tree} is a tree $\T =(V(\T),E(\T))$ with leaf set (also called taxon set) $X=\{1,\ldots,n\} \subset V(\T)$ with only vertices of degree 1 or 3 (internal vertices), where one of the vertices of degree 1 is defined to be the {\it root}, and all edges are directed away from it.  In this paper, when there is no ambiguity we often just write ``phylogenetic tree'' or ``tree'' when referring to a rooted binary phylogenetic tree. Also, when referring to a tree on a leaf set $X$ with $|X|=n$, we write {\it $n$-taxon tree} for short. 

Furthermore, recall that a {\it character} $\f$ is a function $\f: X\rightarrow \C$ for
some set $\C:=\{c_1, c_2, c_3, \ldots, c_r \}$ of $r$ {\it character states} ($r \in \NN$). We denote by $\C^n$ the set of all $r^n$ possible characters on $\C$ and $n$ taxa.
For instance, for the four-state DNA alphabet, $\C_{DNA}=\{A, G, C, T\}$ and the set $\C^n$ consists of $4^n$ elements. An {\it extension} of $\f$ to $V(\T)$ is a map $\g: V(\T)\rightarrow \C$ such that $\g(i) = \f(i)$ for all $i$ in $X$. For such an extension $\g$ of $\f$, we denote by $l_{\T}(\g)$ the number of edges $e=\{u,v\}$ in $\T$ on which a substitution occurs, i.e. where $\g(u) \neq \g(v)$. The {\it parsimony score} of $\f$ on $\T$, denoted by $l_{\T}(\f)$, is obtained by minimizing $l_{\T}(\g)$ over all possible extensions $\g$. Given a tree $\T$ and a character $\f$ on the same taxon set, one can easily calculate the parsimony score of $\f$ on $\T$ with the famous Fitch algorithm \cite{fitch1971}.
Moreover, when a character state changes along one edge of the tree, we refer to this state change as {\it substitution} or {\it mutation}. As for our purposes only so-called manifest mutations are relevant, i.e. those mutations that can be observed and are not reversed, we do not distinguish between mutations and substitutions, which is why we use these terms synonymously.

\subsection{Construction of the OSM matrix}
\label{osm-construction}

We now introduce the OSM framework in a stepwise fashion. The aim of the OSM approach is to determine the effects a single mutation occurring on a rooted tree $\T$ has on a character evolving on that tree.

The first task of this approach is to formalize the term mutation and its effects on a single character state in $\C$. A mutation is an operation $\sigma:\,\C\to\C$ which is bijective, i.e. it satisfies the following condition:
\begin{enumerate}
\item For all $c_i\in\C$ there is a $c_j\in \C$ such that $\sigma(c_i)=c_j$, and if $\sigma(c_i)=\sigma(c_j),$ then $c_i=c_j$.
\end{enumerate}
This guarantees that a mutation affects a character state in a unique fashion. It is well-known that any bijective operation on a finite discrete state set is isomorphic to a {\it permutation} (e.g., \cite{maclane1999}). Therefore, in the following we consider mutations to be permutations.

The next step is to establish which permutations we consider admissible in a model. In other words, we next establish conditions on the set $\mS$ of permutations acting on $\C$.
\begin{enumerate}
\item[2.] For every pair $c_i,c_j\in\C$ there is exactly one operation $\sigma\in\mS$ such that $\sigma(c_i)=c_j$.
\end{enumerate}
This guarantees that every character state change can be observed within a single step and that we do not have any ambiguity. If $\mS$ contains the identity, i.e. the mapping $\sigma_0$ such that $\sigma_0(c_i)=c_i$ for all $c_i\in\C$, then all other permutations in $\mS$ are fix-point free due to Condition 2. Condition 2 also implies that $\mS$ contains exactly $r$ permutations, where $r$ is the number of character states in $\C$. If $\mS$ had more permutations then for all states $c_i\in\C$ there would be a pair of distinct permutations $\sigma_1,\sigma_2\in\mS$ such that $\sigma_1(c_i)=\sigma_2(c_i)$, which would lead to ambiguity. Condition 2 also concludes that we exclude GTR \cite{tavare1986} from the set of admissible models. However, we explain this more in-depth in Section \ref{link_model_perm}.

We add some more useful conditions which give $\mS$ a very convenient structure:
\begin{enumerate}
\item[3.] For all $\sigma_1,\sigma_2\in\mS$ also the product $\sigma_1\circ\sigma_2\in\mS$. In other words, $\mS$ is closed with respect to concatenation of its permutations.
\item[4.] For all $\sigma_1,\sigma_2\in\mS$ we have $\sigma_1\circ\sigma_2=\sigma_2\circ\sigma_1$. Thus, $\mS$ is commutative, and hence the order in which we assign permutations is irrelevant for the outcome.
\item[5.] There is an element $\sigma_0\in\mS$ such that for all $\sigma_1\in\mS$ we have $\sigma_1\circ\sigma_0=\sigma_0\circ\sigma_1=\sigma_1$. As pointed out above, including the identity guarantees that all other permutations will force a state change, a feature which led to the name ``One Step Mutation''.
\item[6.] For every $\sigma_1\in\mS$ there is a $\sigma_2\in\mS$ such that $\sigma_1\circ\sigma_2=\sigma_0$. The existence of such an inverse element guarantees that every operation can be reversed within a single step, which is quite a useful property.
\item[7.] For all $\sigma_1,\sigma_2,\sigma_3\in\mS$ we have $\sigma_1\circ(\sigma_2\circ\sigma_3)=(\sigma_1\circ\sigma_2)\circ\sigma_3=\sigma_1\circ\sigma_2\circ\sigma_3$. Associativity is needed to enforce a group structure on $\mS$.
\end{enumerate}
All of these conditions taken together imply that $\mS$ forms an Abelian group of $r$ permutations. From now on we use the matrix form of permutations for illustration of the operations. A permutation matrix $\sigma$ over $\C$ is represented by an $r\times r$ matrix such that $\sigma_{c_ic_j}=1$ if $\sigma(c_i)=c_j$, and 0 otherwise. In that case, a concatenation ``$\circ$'' is equivalent to the matrix multiplication ``$\cdot$''.

\begin{example}
\label{exa.K3ST_state}
In genetics, the most commonly used character state set is $\C_{\text{DNA}}=\{A,G,C,T\}$. There are two different Abelian groups for four states, namely the Klein-Four-group $\Z_2\times\Z_2$ and the cyclic group $\Z_4$. The Klein-Four-group is constructed from the cyclic group over two elements, identity ($\tau_0$) and flip ($\tau_0$). These take the matrix form
\begin{equation*}
\tau_0=\begin{pmatrix}1&0\\0&1\end{pmatrix},\quad
\tau_1=\begin{pmatrix}0&1\\1&0\end{pmatrix}.
\end{equation*}
The Klein-Four-group consists of the four Kronecker products of these two matrices, i.e. $s_0=\tau_0\otimes\tau_0,\,s_1=\tau_1\otimes\tau_0,\,s_2=\tau_0\otimes\tau_1$, and $s_3=\tau_1\otimes\tau_1$. In full, they take the form:
\begin{equation*}
s_1=\bordermatrix{
	 & A& C& G& T\cr
	A& 0& 0& 1& 0\cr
	C& 0& 0& 0& 1\cr
	G& 1& 0& 0& 0\cr
	T& 0& 1& 0& 0},\quad
s_2=\bordermatrix{
	 & A& C& G& T\cr
	A& 0& 1& 0& 0\cr
	C& 1& 0& 0& 0\cr
	G& 0& 0& 0& 1\cr
	T& 0& 0& 1& 0}, \quad
s_3=\bordermatrix{
	 & A& C& G& T\cr
	A& 0& 0& 0& 1\cr
	C& 0& 0& 1& 0\cr
	G& 0& 1& 0& 0\cr
	T& 1& 0& 0& 0}.
\end{equation*}
This construction coincides with the K3ST model of substitution, where $s_1$ describes {\it transitions} within purines $(A,G)$ and pyrimidines $(C,T)$, $s_2$ represents {\it transversions} within pairs $(A,C)$ and $(G,T)$, and $s_3$ represents the remaining set of transversions within pairs $(A,T)$ and $(C,G)$.

The cyclic group is given by the permutation set
\begin{equation*}
{s'}_1=\bordermatrix{
	& A& C& G& T\cr
	A& 0& 0& 1& 0\cr
	C& 1& 0& 0& 0\cr
	G& 0& 0& 0& 1\cr
	T& 0& 1& 0& 0},\quad
{s'}_2={s'}_1^2={s'}_1\cdot {s'}_1, \quad
{s'}_3={s'}_1^3,\quad {s'}_0={s'}_1^4.
\end{equation*}
Note that the cyclic group $\Z_4$ has a different interpretation with a different ordering of the nucleotides. E.g., our matrix ${s'}_1$ yields the rotation $A\to G\to T\to C\to A$, while Bryant \cite{bryant2005} uses the rotation $A\to C\to G\to T\to A$. 
The cyclic group associated to the latter rotation \cite{bryant2005} is linked to the K2ST substitution model \cite{kimura1980}, where ${s'}_2$ corresponds to the transition within purines and pyrimidines, and ${s'}_1$ and ${s'}_3$ are the (not further distinguished) transversions.
\end{example}

The next step in constructing the OSM matrix is to construct a set $\mS^\T$ of operations over $\C^n$ governed by $\T$, and based on the permutation set $\mS$. To this end, we first define $\mS^n$ as a set of operations which work elementwise, i.e. for $\f=(c_1,\dots,c_n)\in\C^n$ and $\bm{\sigma}\in\mS^n$ we have
\begin{equation*}
\bm{\sigma}(\bm{f}):=(\sigma_1(c_1),\dots,\sigma_n(c_n)),\quad \sigma_i\in\mS.
\end{equation*}
This can also be described by the Kronecker product, i.e. equally
\begin{equation*}
\bm{\sigma}(\bm{f})=\sigma_1\otimes\dots\otimes\sigma_n(\bm{f}).
\end{equation*}
This means that there are $r^n$ different operators in $\mS^n=\mS\otimes\dots\otimes\mS$. Thus, for any pair of characters $\bm{f},\bm{g}\in\C^n$ we can find an operation $\bm{\sigma}\in\mS^n$ such that $\bm{\sigma}(\bm{f})=\bm{g}$.

Another noteworthy consequence of using the Kronecker product is that the elements of $\mS^n$ are permutations over $\C^n$ \cite{horn1991, steeb2011}, and in fact $\mS^n$ satisfies our Conditions 1-7, i.e. $\mS^n$ is an Abelian group over $\C^n$.

In the OSM framework we assume that the permutations acting on a character $\bm{f}\in\C^n$ are derived from the underlying rooted tree $\T$. In particular, regard the pendant edge to a Taxon $1\in X$. A permutation on this edge will only affect the character state of $f_1$. Therefore, the edge implies permutations of type
\begin{equation*}
\bm{\sigma}^{1,i}:=\sigma_i\otimes\sigma_0\dots\otimes\sigma_0,\quad i=1,\dots,r-1.
\end{equation*}
This construction works analogously for all pendant edges, with all but one factor being the identity $\sigma_0$, while the one non-identity factor is one of the remaining $r-1$ permutations in $\mS$. A permutation at an interior edge affects the character states of all its descendants, i.e. those taxa whose path to the root passes that edge. E.g., assume Taxa 1 and 2 form a cherry, i.e. their most recent common ancestor has no other descendants, and permutation $\sigma_i\in\mS,\,i=1,\dots,r-1$ is acting on the edge leading to this ancestor. Then, we get the permutation
\begin{equation*}
\bm{\sigma}^{12,i}:=\sigma_i\otimes\sigma_i\otimes\sigma_0\dots\otimes\sigma_0=\bm{\sigma}^{1,i}\otimes\bm{\sigma}^{2,i}.
\end{equation*}
The right hand side equation shows that a single permutation on an internal edge has the same effect as simultaneously applying the same permutation on the pendant edges of all descendant taxa. This also shows that the set $\mS^X$ of all permutations on the pendant edges is a generator of $\mS^n$, i.e. the closure of $\mS^X$ contains all permutations in $\mS^n$. Since $\mS^n$ contains a single permutation to transform character $\bm{f}\in\C^n$ into $\bm{g}\in\C^n$, and since $\mS^X$ generates $\mS^n$, there is a shortest chain of permutations in $\mS^X$ which transforms $\bm{f}$ into $\bm{g}$. $\mS^X$ is also the set of permutations implied by the star tree for $X$. For every $X$-tree $\T$ we have $\mS^\T\supseteq \mS^X$, and therefore $\mS^\T$ is a generator for $\mS^n$, too. 
An illustration of such a generator set $\mS^T$ over the character set $\C^n$ is the so-called {\it Cayley graph} \cite{cayley1878}, which has as vertices the elements of $\C^n$, and two elements $\f,\g\in\C^n$ are connected if there is a permutation $\bm{\sigma}\in\mS^\T$ such that $\bm{\sigma}(\f)=\g$. In \cite{klaere2008} Cayley graphs have been presented as alternative illustrations of the tree $\T$ over a binary state set $\C$.

\begin{example}
\label{exa.K3ST_character}
Regard the K3ST model from \ref{exa.K3ST_state} and the rooted two-taxon tree depicted in \figcite{\figOSM}(a). With this $\mS^\T_{K3ST}$ is given by the set
\begin{align*}
\bm{s}^{e_1,1}&:=s_1\otimes s_0,&\bm{s}^{e_2,1}&:=s_0\otimes s_1,&\bm{s}^{e_{12},1}&:=s_1\otimes s_1,\\
\bm{s}^{e_1,2}&:=s_2\otimes s_0,&\bm{s}^{e_2,2}&:=s_0\otimes s_2,&\bm{s}^{e_{12},2}&:=s_2\otimes s_2,\\
\bm{s}^{e_1,3}&:=s_3\otimes s_0,&\bm{s}^{e_2,3}&:=s_0\otimes s_3,&\bm{s}^{e_{12},3}&:=s_3\otimes s_3.
\end{align*}
Each operation is thus a symmetric $16\times16$ permutation matrix depicting a transition ($\bm{s}^{e,1}$), transversion 1 ($\bm{s}^{e,2}$), or transversion 2 ($\bm{s}^{e,3}$) along edge $e\in E(\T)$. \figcite{\figOSM}(b), (c), and (d) display the permutation matrices for a transition on branch $e_1$ ($\bm{s}^{e_1,1}$), $e_2$ ($\bm{s}^{e_2,1}$) and $e_{12}$ ($\bm{s}^{e_{12},1}$), respectively. 
\figcite{\figCayley}(a) shows the Cayley graph associated with $\mS^\T_{K3ST}$.
\end{example}

We are now in a position to recall the definition of the {\it OSM matrix} $M_\T$ for a rooted binary phylogenetic tree $\T$ as explained in \cite{klaere2008} and \cite{nguyen2011b}. For an edge $e\in E(\T)$ we denote by $p_e$ the relative branch length of $e$, i.e. its actual length divided by the length of $\T$. Thus, one can view $p_e$ as the probability that a mutation is observed at edge $e$ given that a mutation occurred on $\T$. Clearly, $\sum_{e\in E(\T)}p_e=1$. Further, denote by $\alpha_{e,i}$ the probability that this mutation on $e$ is of type $i\in\{1,\dots,r-1\}$ with $\sum_{i\in 1}^{r-1}\alpha_{e,i}=1$ for all $e\in E(\T)$. Then the OSM matrix is the convex sum of the elements in $\mS^\T$, where each permutation $\bm{\sigma}^{e,i}$ is multiplied by $p_e\alpha_{e,i}$, the probability of hitting the edge $e$ with permutation $\sigma_i\in\mS$. Thus, we obtain:
\begin{equation}\label{eq:osm-matrix}
M_\T = \sum_{e \in E(\T)}\sum_{i=1}^{r-1}\alpha_{e,i} p_e \bm{\sigma}^{e,i}.
\end{equation}
$M_\T$ can be regarded as the weighted exchangeability matrix for all characters given that a substitution occurs somewhere on the tree $\T$. \figcite{\figOSM}(e) depicts the OSM matrix for the tree in \figcite{\figOSM}(a). Here, colors indicate relative branch lengths $p_e$, and patterns denote permutation types $\alpha_i$. E.g., a blue square with horizontal lines indicates the product $p_{e_2}\alpha_{e_2,1}$, i.e. the probability of observing a Transition $s_1$ on Edge $e_2$.

\subsection{The transformation problem}
\label{problem}

With the construction of $\mS^\T$ we have generated the tools needed to formally describe the computations in Step 4 of the {\sc Misfits} method introduced by Nguyen {\it et al.} \cite{nguyen2011}. Given a rooted tree $\T$ and two characters $\f$ and $\f^d$ in $\C^n$, we want to compute the minimal number of substitutions required on the tree to convert $\f$ into $\f^d$.

In our framework this corresponds to finding the smallest number $k$ of permutations $\bm{\sigma}_1,\dots,\bm{\sigma}_k\in\mS^\T$ such that $\bm{\sigma}_1\circ\dots\circ\bm{\sigma}_k(\f)=\f^d$. The number $k$ has multiple equivalent interpretations. It is also the length of the shortest path between $\f$ and $\f^d$ in the Cayley graph for $\mS^T$, where this path corresponds exactly to the chain $\bm{\sigma}_1\circ\dots\circ\bm{\sigma}_k$ since each edge in the Cayley graph corresponds to an operation in $\mS^\T$. $k$ is also the smallest matrix power such that $M_\T^j=0$ for $j<k$ and $M_\T^k>0$, because a positive entry in $M_\T^k$ means that there is a concatenation of $k$ permutations connecting the associated characters.

Nguyen et al. \cite{nguyen2011} presented an efficient procedure to compute the minimal number of  substitutions as summarized in Algorithm \ref{alg}, and we prove its correctness in Section \ref{parsimonyproof}.

\begin{algorithm}\label{alg}\mbox{}
{\bf INPUT:} rooted binary phylogenetic tree $\T$ on leaf set $X$, characters $\bm{f}$ and $\f^d$ on $X$, group $\Sigma$.\par
{\bf ITERATION 1:} Align characters $\bm{f}$ and $\f^d$ and find the corresponding substitution type $\sigma_i$ which translates $f_j$ into $\f^d_j$ for all positions $j=1,\ldots,|X|$. Let $\bm{\sigma} \in \Sigma^n$ be the resulting operation. \par
{\bf ITERATION 2:} Let $\bm{h}:=c_1\ldots c_1$ be the constant character on $X$ with some $c_1 \in \mathcal{C}$ on all positions. Apply $\bm{\sigma}$ to $h$ and call the derived character $\bm{c}$.\par
{\bf ITERATION 3:} Calculate $m:=l_{\T}(\bm{c})$. \par
{\bf OUTPUT:} Minimum number $m$ of  substitutions needed to evolve $\f^d$ instead of $\bm{f}$ on $\T$.\par
\end{algorithm}

\begin{example}
\label{exampleK3ST}
\figcite{\figEgKlein} demonstrates how Algorithm \ref{alg} works under the K3ST model, i.e. when the group is $\Sigma=\Sigma_{K3ST}$ (\figcite{\figEgKlein}(a)).
Consider the rooted five-taxon tree in \figcite{\figEgKlein}(b) and the character $GTAGA$ at the leaves. Assume that the character $GTAGA$ is to be converted into character $ACCTC$. By comparing the two characters positionwise, we need a substitution $s_1$ on the external branch leading to Taxon 1 to convert $G$ into $A$ at the first position.
Similarly, we need a substitution $s_1$ on the external branch leading to Taxon 2, and a substitution $s_2$ on every external branch leading to Taxa 3, 4, and 5. Thus, the operation $\bm{s}:=(s_1, s_1, s_2, s_2, s_2)$ transfers the character $GTAGA$ into the character $ACCTC$. As the operation $\bm{s}:= (s_1, s_1, s_2, s_2, s_2)$ also translates the constant character $AAAAA$ into $GGCCC$, converting $GTAGA$ into $ACCTC$ is equivalent to evolving the character state $A$ at the root along the tree to obtain the character $GGCCC$ at the leaves. The Fitch algorithm applied to the character $GGCCC$ with the constraint that the character state at the root is $A$ produces a unique most parsimonious solution of two substitutions as depicted by \figcite{\figEgKlein}(c).
\end{example}

\section{Results}
\label{results}
\bigskip

\subsection{The impact of parsimony on the estimation of  substitutions. }\label{parsimonyproof}
\bigskip
In this section, we provide some mathematical insights into the role of Maximum Parsimony in the estimation of the number of  substitutions needed to convert a character into another one as explained in Section \ref{problem}. In particular, we deliver a proof for Algorithm \ref{alg}. 

\begin{theorem} Let $\mathcal{T}$ be a rooted binary phylogenetic tree on taxon set $X$ and let $\f$ be a character that evolved on $\mathcal{T}$ due to some evolutionary model and let $\f^d$ be another character on $X$. Then, the number of  substitutions to be put on $\mathcal{T}$ which change the evolution of $\f$ in such a way that $\f^d$ evolves instead of $\f$ can be calculated with Algorithm $\ref{alg}$.
\end{theorem}

\begin{proof} Let $\f$, $\f^d$, $X$, $\T$ and $\mS$ be as required for the input of Algorithm \ref{alg}. Then, the number of  substitutions needed to evolve $\f^d$ on $\T$ rather than $\f$ depends solely on operation $\bm{\sigma}$. In order to see this, note that $\bm{\sigma}$ describes an explicit way to translate $\f$ into $\f^d$ step by step, i.e. for each taxon seperately. The basic idea now is that in order to minimize the number of required substitutions, we need to consider the underlying tree $\T$, as this may allow a single substitution to act on an ancestor of taxa that undergo the same substitution type rather than on each taxon separately. This idea has been described above in Section \ref{osm-construction}, and it coincides precisely with the idea of the parsimony principle.

However, in order to avoid confusion regarding the operation $\bm{\sigma}$ as a character on which to apply parsimony, Algorithm \ref{alg} instead acts on the constant character. Clearly, in order to evolve the constant character $\bm{h}:=c_1\cdots c_1$ on a tree with root state $c_1$, the corresponding operation would be $\tilde{\bm{\sigma}}:=\sigma_0\circ\cdots\circ\sigma_0$. If instead of $\tilde{\bm{\sigma}}$ we let $\bm{\sigma}$ act on $\bm{h}$, the resulting character $\bm{c}$ will differ from $\bm{h}$ in the same way $\f^d$ differs from $\f$. Note that two character states in $\bm{c}$ are identical if and only if the corresponding substitutions in $\bm{\sigma}$ are identical, too. Therefore, it is possible to let MP act on $\bm{c}$ rather than directly on $\bm{\sigma}$.

By the definition of Maximum Parsimony, when applied to $\bm{c}$ on tree $\T$ with given root state $c_1$, it calculates the minimum number $m$ of substitutions to explain $\bm{c}$ on $\T$. This number $m$ is therefore precisely the number of substitions needed to generate $\bm{c}$ on $\T$ rather than $\bm{h}$. But as explained before, $\bm{c}$ and $\bm{h}$ are by definition related the same way as $\f$ and $\f^d$. Therefore, $m$ also is the number of substitutions needed to generate $\f^d$ on $\T$ rather than $\f$. This completes the proof.
\end{proof}

\subsection{The impact of different groups}
For any alphabet, there might be more than one Abelian group. Different groups might result in different numbers of  substitutions required to translate a character into another character. We illustrate this in the following examples. For the four-state nucleotide alphabet there are two Abelian groups, namely the Klein-four group and the cyclic group (see above). The cyclic group $\mS_{c}$ consists of the identity matrix $s'_0$ and the three substitution types $s'_1, s'_2, s'_3$ depicted by \figcite{\figEgCyclic}(a). Hence, $\mS_c = \{s'_0, s'_1, s'_2, s'_3\}$. Under $\mS_{c}$, we note that a substitution type which changes a character state $c_i$ to $c_j$ does not necessarily change $c_j$ to $c_i$.

\begin{example}
\label{exampleCyclic}
Assume the rooted five-taxon tree in \figcite{\figEgCyclic}(b) and the character $GTAGA$ at the leaves, which is to be converted into character $ACCTC$. The tree and the two characters are the same as in Example \ref{exampleK3ST}. By comparing the two characters positionwise under the group $\mS_c$, we need a substitution $s'_3$ (depicted in blue in \figcite{\figEgCyclic}(a)) on the external branch leading to Taxon 1 to convert $G$ into $A$ at the first position. Analogously, we need a substitution $s'_1$ on the external branches leading to Taxon 2 and to Taxon 4, and a substitution $s'_3$ on the external branches leading to Taxon 3 and to Taxon 5. Thus, the operation $\bm{s}' := (s'_3, s'_1, s'_3, s'_1, s'_3)$ transfers the character $GTAGA$ into the character $ACCTC$. As the operation $\bm{s}'$ also translates the constant character $AAAAA$ into $CGCGC$, converting $GTAGA$ into $ACCTC$ is equivalent to evolving the character state $A$ at the root along the tree to obtain the character $CGCGC$ at the leaves. The Fitch algorithm applied to the character $CGCGC$ with the constraint that the character state at the root is $A$ produces a unique most parsimonious solution of three substitutions as depicted by \figcite{\figEgCyclic}(c). Thus, under the $\mS_c$ group we need one  substitution more than under the $\mS_{K3ST}$ group.
\end{example}

Note that variation of the minimum number of substitutions needed to translate a character into another one between different groups is not surprising: As different substitution types are needed to translate one pattern into the other one, depending solely on the underlying group, one group might need the same substitution type for some neighboring branches in the tree and another group different ones. Informally speaking, this would imply that in the first case, the substitution could be ``pulled up'' by the Fitch algorithm to happen on an ancestral branch, whereas in the second case this would not be possible.

\subsection{The link between substitution models and permutation matrices} \label{link_model_perm}

In Examples \ref{exa.K3ST_state} and \ref{exa.K3ST_character} we have shown that the K3ST substitution model can be included into our framework. This section aims at discussing alternative models and how to identify their use (or lack thereof) for our approach. The set $\mS^\T$ contains a set of permutations which act on the characters in $\C^n$. 

Most substitution models assume the independence of the different branches of a tree to compute the joint probability of the characters in $\C^n$. Therefore, they use the probabilities for substitutions among the character states in $\C$ along the edges of the tree $\T$. We now establish a probabilistic link between $\mS^\T$ and $\C^n$. This link is provided by Birkhoff's theorem:
\begin{theorem}[Birkhoff's theorem, e.g., \cite{horn1990}, Theorem 8.7.1] A matrix $M$ is doubly stochastic, i.e., each column and each row of $M$ sum to 1, if and only if for some $N<\infty$ there are permutation matrices $\sigma_1,\dots,\sigma_N$ and positive scalars $\alpha_1,\dots,\alpha_N\in\real$ such that $\alpha_1+\dots+\alpha_n=1$ and $M=\alpha_1\sigma_1+\dots+\alpha_N\sigma_N$.
\end{theorem}
Therefore, the weighted sum of the permutation matrices in $\mS^\T$ yields a doubly stochastic matrix $M_\T$ as introduced in Section \ref{osm-construction}. $M_\T$ also describes a random walk on $\C^n$ governed by $\T$ where the single step in $\C^n$ is illustrated by the associated Cayley graph. Its stationary distribution is uniform, i.e. when we throw sufficiently many mutations on $\T$ then we expect to see each pattern with probability $1/r^n$.

Another, even more useful consequence of Birkhoff's theorem is the fact that it tells us which substitution models are suited for the OSM approach. If the transition matrix associated with the model is doubly stochastic, then we find a set of permutations which give rise to the model.

Let us see how this influences the symmetric form of the general time reversible model (GTR). It has the transition matrix
\begin{equation*}
P_{GTR}=\bordermatrix{
	 & A& C& G& T\cr
	A& 1-a-b-c& a& b& c\cr
	C& a& 1-a-d-e& d& e\cr
	G& b& d& 1-b-d-f& f\cr
	T& c& e& f& 1-c-e-f}. \quad
\end{equation*}
Assigning permutation matrices to the respective parameters yields the set $\mS_{GTR}$ with elements $s_0$ (identity) and 
\begin{align*}
s_a&=
\begin{pmatrix}
0&1&0&0\\1&0&0&0\\0&0&1&0\\0&0&0&1
\end{pmatrix},&
s_b&=
\begin{pmatrix}
0&0&1&0\\0&1&0&0\\1&0&0&0\\0&0&0&1
\end{pmatrix},&
s_c&=
\begin{pmatrix}
0&0&0&1\\0&1&0&0\\0&0&1&0\\1&0&0&0
\end{pmatrix}\\
s_d&=
\begin{pmatrix}
1&0&0&0\\0&0&1&0\\0&1&0&0\\0&0&0&1
\end{pmatrix},&
s_e&=
\begin{pmatrix}
1&0&0&0\\0&0&0&1\\0&0&1&0\\0&1&0&0
\end{pmatrix},&
s_f&=
\begin{pmatrix}
1&0&0&0\\0&1&0&0\\0&0&0&1\\0&0&1&0
\end{pmatrix}.
\end{align*}
The weighted sum of the non-identity elements yields 
\begin{equation*}
a s_a+b s_b+c s_c+d s_d+e s_e+f s_f=
\begin{pmatrix}
d+e+f&a&b&c\\
a&b+c+f&d&e\\
b&d&a+c+e&f\\
c&e&f&a+b+d
\end{pmatrix},
\end{equation*}
which is equal to $P_{GTR}$ because $a+b+c+d+e+f=1$. Thus, the set $\mS_{GTR}$ is to GTR what $\mS_{K3ST}$ is to K3ST. However, $\mS_{GTR}$ does not satisfy Condition 2, because it contains more than four elements. Therefore, it creates ambiguity since for each nucleotide there are three permutations which do not change the nucleotide. It is also not commutative (Condition 4) which means the order in which we assign the permutations matters. And it is not closed under matrix multiplication (Condition 3), which means that a concatenation of permutations in $\mS_{GTR}$ might lead to a new permutation not in $\mS_{GTR}$, i.e. we would encounter a new mutation type. All of this shows why the overall applicability of GTR to the OSM approach is rather limited. More complex models like Tamura-Nei \cite{tamura1993} do not even permit the decomposition of its transition matrix into the convex sum of permutation matrices. However, including the concept of partial permutation matrices \cite{horn1991} can address this problem. While this approach is interesting for future work, it is beyond the scope of this paper.

\subsection{Application to other biologically interesting sets}

As stated in Section \ref{osm-construction}, the OSM model only requires an underlying Abelian group. Thus, the OSM setting is applicable not only to binary data or four-state (DNA or RNA) data, but also to doublet, codon, and amino acid characters.

In particular, there are four Abelian groups for the twenty-state amino acid alphabet, namely $\Z_2\times\Z_2\times\Z_5,\,\Z_4\times\Z_5,\,\Z_2\times\Z_{10}$, and the cyclic group $\Z_{20}$ (see e.g.,\cite{keilen2001} for a complete list of all groups with up to 35 elements). Their construction is analogous to the construction of the Klein-Four group in Example \ref{exa.K3ST_state}. For example, the elements of $\Z_4\times\Z_5$ are Kronecker products of one of the four permutations in the cyclic group $\Z_4$ with one of the five permutations of the cyclic group $\Z_5$. 

\figcite{\figAminoAcid} depicts the 20 substitution types, i.e. the 20 operations including the identity acting on the amino acid character states for all four Abelian groups. If we assign probabilities to the substitution types in the matrices, the resulting matrices are doubly stochastic. The matrices show several features of the groups, e.g. that contrary to the Klein-Four group the elements of the group are not self-inverse but instead the effect of a permutation is reversed by a different mutation. Such events are present in some models of nucleotide evolution, like the strand symmetric model \cite{casanellas2005}, and relatively common in amino acid models where the transition matrix is generated by, e.g., counting mutation types in amino acid alignments (see, e.g., \cite{kosiol2005} for an overview). It might be interesting to see whether any of these can be fitted. The illustrations in \figcite{\figAminoAcid} also suggest some ordering of the amino acids to fit the model. For instance, $\Z_2\times\Z_2\times\Z_5$ and $\Z_4\times\Z_5$ seem to partition the sets into four groups with five elements each.

\section*{Conclusions} 

In this paper, we provide the necessary mathematical background for the OSM setting which was introduced and used previously \cite{nguyen2011,nguyen2011b}, but had not been analyzed mathematically for more than two character states. Moreover, the present paper also delivers new insight concerning the requirements for the OSM model to work: In fact, we were able to show that mathematically, it is sufficient to have an underlying Abelian group -- which shows a generalization of the OSM concept that was believed to be impossible previously \cite{nguyen2011}. Therefore, we show that OSM is applicable to any number of states.

However, note that the original intuition of the authors in \cite{nguyen2011} was biologically motivated: The authors supposed that the group not only has to be Abelian, but also symmetric in the sense that each operation can be undone by being applied a second time. Thinking about the DNA, for instance, this works: For example, the transition from A to G can be reverted by another substitution of the same type, namely a transition from G to A. This symmetry criterion is fulfilled by the Klein-Four group, but not by the cyclic group on four states. Unfortunately, for 20 states there is no Abelian group fulfilling this criterion, which is why the demonstrated generalization to 20 states does not provide a nice symmetry (r.f. \figcite{\figAminoAcid}). Therefore, it remains unclear at this stage if there are biologically motivated settings for which our twenty-state generalization is directly applicable. 

\bigskip

\section*{Authors' contributions}
    All authors contributed equally.
\bigskip
\section*{Acknowledgements}
\ifthenelse{\boolean{publ}}{\small}{}
SK thanks Marston Conder for fruitful discussions on the group theoretical background.
This work is financially supported by the Wiener Wissenschafts-, Forschungs- and Technologiefonds (WWTF).
AvH also acknowledges the funding from the DFG Deep Metazoan Phylogeny project, SPP (HA1628/9) and the support from the Austrian GEN-AU project Bioinformatics Integration Network III.
   

\newpage
{\ifthenelse{\boolean{publ}}{\footnotesize}{\small}
 \bibliographystyle{bmc_article}  
  \bibliography{misfitsTheoryBib} }     


\begin{thebibliography}{10}
\providecommand{\url}[1]{[#1]}
\providecommand{\urlprefix}{}

\bibitem{klaere2008}
Klaere S, Gesell T, von Haeseler A: \textbf{The impact of single substitutions
  on multiple sequence alignments}. \emph{Philos. Trans. R. Soc. Lond., B,
  Biol. Sci.} 2008, \textbf{363}:4041--4047.

\bibitem{nguyen2011}
Nguyen MAT, Klaere S, von Haeseler A: \textbf{{{M}{I}{S}{F}{I}{T}{S}:
  evaluating the goodness of fit between a phylogenetic model and an
  alignment}}. \emph{Mol. Biol. Evol.} 2011, \textbf{28}:143--152.

\bibitem{durbin1998}
Durbin R, Eddy SR, Krogh A, Mitchison G: \emph{Biological sequence analysis -
  Probabilistic models of proteins and nucleic acids}. Cambridge: Cambridge
  University Press 1998.

\bibitem{mount2004}
Mount DW: \emph{Bioinformatics{:} Sequence and Genome Analysis}. New York: Cold
  Spring Harbor 2004.

\bibitem{semple2003}
Semple C, Steel M: \emph{Phylogenetics.} New York: Oxford University Press
  2003.

\bibitem{kimura1981}
Kimura M: \textbf{Estimation of Evolutionary Distances between Homologous
  Nucleotide Sequences}. \emph{Proc. Natl. Acad. Sci. USA} 1981,
  \textbf{78}:454--458.

\bibitem{fitch1971}
Fitch WM: \textbf{Toward defining the course of evolution: Minimum change for a
  specific tree topology}. \emph{Syst. Zool.} 1971, \textbf{20}:406--416.

\bibitem{humphreys1996}
Humphreys JF: \emph{A course in group theory}. New York: Oxford University
  Press 1996.

\bibitem{keilen2001}
Keilen T: \textbf{Endliche Gruppen. Eine Einf\"uhrung mit dem Ziel der
  Klassifikation von Gruppen kleiner Ordnung} 2000,
  \urlprefix\url{[http://www.mathematik.uni-kl.de/~wwwagag/download/scripts/En%
dliche.Gruppen.pdf]}.

\bibitem{maclane1999}
MacLane S, Birkhoff G: \emph{Algebra}. American Mathematical Society 1999.

\bibitem{tavare1986}
Tavar{\'e} S: \textbf{Some probabilistic and statistical problems on the
  analysis of {DNA} sequences}. \emph{Lec. Math. Life Sci.} 1986,
  \textbf{17}:57--86.

\bibitem{bryant2005}
Bryant D: \textbf{Extending Tree Models to Split Networks}. In \emph{Algebraic
  Statistics for Computational Biology}. Edited by Pachter L, Sturmfels B,
  Cambridge University Press 2005.

\bibitem{kimura1980}
Kimura M: \textbf{A Simple Method for Estimating Evolutionary Rates of Base
  Substitutions through Comparative Studies of Nucleotide Sequences}. \emph{J.
  Mol. Evol.} 1980, \textbf{16}:111--120.

\bibitem{horn1991}
Horn RA, Johnson CR: \emph{Topics in matrix analysis}. New York: Oxford
  University Press 1991.

\bibitem{steeb2011}
Steeb WH, Hardy Y: \emph{Matrix Calculus and Kronecker Product: A Practical
  Approach to Linear and Multilinear Algebra}. Singapore: World Scientific
  Publishing, 2 edition 2011.

\bibitem{cayley1878}
Cayley A: \textbf{Desiderata and Suggestions: No. 2. The Theory of Groups:
  Graphical Representation}. \emph{American Journal of Mathematics} 1878,
  \textbf{1}(2):pp. 174--176.

\bibitem{nguyen2011b}
Nguyen MAT, Gesell T, von Haeseler A: \textbf{ImOSM: Intermittent Evolution and
  Robustness of Phylogenetic Methods}. \emph{Mol. Biol. Evol.} 2011.

\bibitem{horn1990}
Horn RA, Johnson CR: \emph{Matrix analysis}. Cambridge University Press 1990.

\bibitem{tamura1993}
Tamura K, Nei M: \textbf{Estimation of the number of nucleotide substitutions
  in the control region of mitochondrial DNA in humans and chimpanzees.}
  \emph{Molecular Biology and Evolution} 1993, \textbf{10}(3):512--526.

\bibitem{casanellas2005}
Casanellas M, Sullivant S: \textbf{The Strand Symmetric Model}. In
  \emph{Algebraic Statistics for Computational Biology}. Edited by Pachter L,
  Sturmfels B, Cambridge University Press 2005.

\bibitem{kosiol2005}
Kosiol C, Goldman N: \textbf{Different Versions of the Dayhoff Rate Matrix}.
  \emph{Molecular Biology and Evolution} 2005, \textbf{22}(2):193--199.

\end{thebibliography}

\newcommand{\BMCxmlcomment}[1]{}

\BMCxmlcomment{

<refgrp>

<bibl id="B1">
  <title><p>The impact of single substitutions on multiple sequence
  alignments</p></title>
  <aug>
    <au><snm>Klaere</snm><fnm>S</fnm></au>
    <au><snm>Gesell</snm><fnm>T</fnm></au>
    <au><snm>Haeseler</snm><fnm>A</fnm></au>
  </aug>
  <source>Philos. Trans. R. Soc. Lond., B, Biol. Sci.</source>
  <pubdate>2008</pubdate>
  <volume>363</volume>
  <fpage>4041</fpage>
  <lpage>4047</lpage>
</bibl>

<bibl id="B2">
  <title><p>{{M}{I}{S}{F}{I}{T}{S}: evaluating the goodness of fit between a
  phylogenetic model and an alignment}</p></title>
  <aug>
    <au><snm>Nguyen</snm><fnm>M. A. T.</fnm></au>
    <au><snm>Klaere</snm><fnm>S.</fnm></au>
    <au><snm>Haeseler</snm><fnm>A.</fnm></au>
  </aug>
  <source>Mol. Biol. Evol.</source>
  <pubdate>2011</pubdate>
  <volume>28</volume>
  <fpage>143</fpage>
  <lpage>-152</lpage>
</bibl>

<bibl id="B3">
  <title><p>Biological sequence analysis - Probabilistic models of proteins and
  nucleic acids</p></title>
  <aug>
    <au><snm>Durbin</snm><fnm>R</fnm></au>
    <au><snm>Eddy</snm><fnm>SR</fnm></au>
    <au><snm>Krogh</snm><fnm>A</fnm></au>
    <au><snm>Mitchison</snm><fnm>G</fnm></au>
  </aug>
  <publisher>Cambridge: Cambridge University Press</publisher>
  <pubdate>1998</pubdate>
</bibl>

<bibl id="B4">
  <title><p>Bioinformatics{:} Sequence and Genome Analysis</p></title>
  <aug>
    <au><snm>Mount</snm><fnm>DW</fnm></au>
  </aug>
  <publisher>New York: Cold Spring Harbor</publisher>
  <pubdate>2004</pubdate>
</bibl>

<bibl id="B5">
  <title><p>Phylogenetics.</p></title>
  <aug>
    <au><snm>Semple</snm><fnm>C.</fnm></au>
    <au><snm>Steel</snm><fnm>M.</fnm></au>
  </aug>
  <publisher>New York: Oxford University Press</publisher>
  <pubdate>2003</pubdate>
</bibl>

<bibl id="B6">
  <title><p>Estimation of Evolutionary Distances between Homologous Nucleotide
  Sequences</p></title>
  <aug>
    <au><snm>Kimura</snm><fnm>M</fnm></au>
  </aug>
  <source>Proc. Natl. Acad. Sci. USA</source>
  <pubdate>1981</pubdate>
  <volume>78</volume>
  <fpage>454</fpage>
  <lpage>458</lpage>
</bibl>

<bibl id="B7">
  <title><p>Toward defining the course of evolution: Minimum change for a
  specific tree topology</p></title>
  <aug>
    <au><snm>Fitch</snm><fnm>WM</fnm></au>
  </aug>
  <source>Syst. Zool.</source>
  <pubdate>1971</pubdate>
  <volume>20</volume>
  <fpage>406</fpage>
  <lpage>416</lpage>
</bibl>

<bibl id="B8">
  <title><p>A course in group theory</p></title>
  <aug>
    <au><snm>Humphreys</snm><fnm>JF</fnm></au>
  </aug>
  <publisher>New York: Oxford University Press</publisher>
  <pubdate>1996</pubdate>
  <fpage>289pp</fpage>
</bibl>

<bibl id="B9">
  <title><p>Endliche Gruppen. Eine Einf\"uhrung mit dem Ziel der Klassifikation
  von Gruppen kleiner Ordnung</p></title>
  <aug>
    <au><snm>Keilen</snm><fnm>T</fnm></au>
  </aug>
  <pubdate>2000</pubdate>
  <url>http://www.mathematik.uni-kl.de/~wwwagag/download/scripts/Endliche.Grup%
pen.pdf</url>
</bibl>

<bibl id="B10">
  <title><p>Algebra</p></title>
  <aug>
    <au><snm>MacLane</snm><fnm>S</fnm></au>
    <au><snm>Birkhoff</snm><fnm>G</fnm></au>
  </aug>
  <publisher>American Mathematical Society</publisher>
  <pubdate>1999</pubdate>
</bibl>

<bibl id="B11">
  <title><p>Some probabilistic and statistical problems on the analysis of
  {DNA} sequences</p></title>
  <aug>
    <au><snm>Tavar{\'e}</snm><fnm>S</fnm></au>
  </aug>
  <source>Lec. Math. Life Sci.</source>
  <pubdate>1986</pubdate>
  <volume>17</volume>
  <fpage>57</fpage>
  <lpage>86</lpage>
</bibl>

<bibl id="B12">
  <title><p>Extending Tree Models to Split Networks</p></title>
  <aug>
    <au><snm>Bryant</snm><fnm>D.</fnm></au>
  </aug>
  <source>Algebraic Statistics for Computational Biology</source>
  <publisher>Cambridge University Press</publisher>
  <editor>Pachter, L. and Sturmfels, B.</editor>
  <section><title><p>17</p></title></section>
  <pubdate>2005</pubdate>
</bibl>

<bibl id="B13">
  <title><p>A Simple Method for Estimating Evolutionary Rates of Base
  Substitutions through Comparative Studies of Nucleotide Sequences</p></title>
  <aug>
    <au><snm>Kimura</snm><fnm>M</fnm></au>
  </aug>
  <source>J. Mol. Evol.</source>
  <pubdate>1980</pubdate>
  <volume>16</volume>
  <fpage>111</fpage>
  <lpage>120</lpage>
</bibl>

<bibl id="B14">
  <title><p>Topics in matrix analysis</p></title>
  <aug>
    <au><snm>Horn</snm><fnm>RA</fnm></au>
    <au><snm>Johnson</snm><fnm>CR</fnm></au>
  </aug>
  <publisher>New York: Oxford University Press</publisher>
  <pubdate>1991</pubdate>
  <fpage>607pp</fpage>
</bibl>

<bibl id="B15">
  <title><p>Matrix Calculus and Kronecker Product: A Practical Approach to
  Linear and Multilinear Algebra</p></title>
  <aug>
    <au><snm>Steeb</snm><fnm>WH</fnm></au>
    <au><snm>Hardy</snm><fnm>Y</fnm></au>
  </aug>
  <publisher>Singapore: World Scientific Publishing</publisher>
  <edition>2</edition>
  <pubdate>2011</pubdate>
</bibl>

<bibl id="B16">
  <title><p>Desiderata and Suggestions: No. 2. The Theory of Groups: Graphical
  Representation</p></title>
  <aug>
    <au><snm>Cayley</snm><fnm>A</fnm></au>
  </aug>
  <source>American Journal of Mathematics</source>
  <publisher>The Johns Hopkins University Press</publisher>
  <pubdate>1878</pubdate>
  <volume>1</volume>
  <issue>2</issue>
  <fpage>pp.174</fpage>
  <lpage>176</lpage>
</bibl>

<bibl id="B17">
  <title><p>ImOSM: Intermittent Evolution and Robustness of Phylogenetic
  Methods</p></title>
  <aug>
    <au><snm>Nguyen</snm><fnm>M. A. T.</fnm></au>
    <au><snm>Gesell</snm><fnm>T.</fnm></au>
    <au><snm>Haeseler</snm><fnm>A.</fnm></au>
  </aug>
  <source>Mol. Biol. Evol.</source>
  <pubdate>2011</pubdate>
</bibl>

<bibl id="B18">
  <title><p>Matrix analysis</p></title>
  <aug>
    <au><snm>Horn</snm><fnm>RA</fnm></au>
    <au><snm>Johnson</snm><fnm>CR</fnm></au>
  </aug>
  <publisher>Cambridge University Press</publisher>
  <pubdate>1990</pubdate>
</bibl>

<bibl id="B19">
  <title><p>Estimation of the number of nucleotide substitutions in the control
  region of mitochondrial DNA in humans and chimpanzees.</p></title>
  <aug>
    <au><snm>Tamura</snm><fnm>K</fnm></au>
    <au><snm>Nei</snm><fnm>M</fnm></au>
  </aug>
  <source>Molecular Biology and Evolution</source>
  <pubdate>1993</pubdate>
  <volume>10</volume>
  <issue>3</issue>
  <fpage>512</fpage>
  <lpage>526</lpage>
</bibl>

<bibl id="B20">
  <title><p>The Strand Symmetric Model</p></title>
  <aug>
    <au><snm>Casanellas</snm><fnm>M.</fnm></au>
    <au><snm>Sullivant</snm><fnm>S.</fnm></au>
  </aug>
  <source>Algebraic Statistics for Computational Biology</source>
  <publisher>Cambridge University Press</publisher>
  <editor>Pachter, L. and Sturmfels, B.</editor>
  <section><title><p>16</p></title></section>
  <pubdate>2005</pubdate>
</bibl>

<bibl id="B21">
  <title><p>Different Versions of the Dayhoff Rate Matrix</p></title>
  <aug>
    <au><snm>Kosiol</snm><fnm>C</fnm></au>
    <au><snm>Goldman</snm><fnm>N</fnm></au>
  </aug>
  <source>Molecular Biology and Evolution</source>
  <pubdate>2005</pubdate>
  <volume>22</volume>
  <issue>2</issue>
  <fpage>193</fpage>
  <lpage>199</lpage>
</bibl>

</refgrp>
} 


\ifthenelse{\boolean{publ}}{\end{multicols}}{}



\section*{Figures}

\bigskip
\subsection*{Figure \figOSM - Construction of the OSM matrix}

\begin{figure}[!htb]
  \center
  \includegraphics[width=0.7\textwidth]{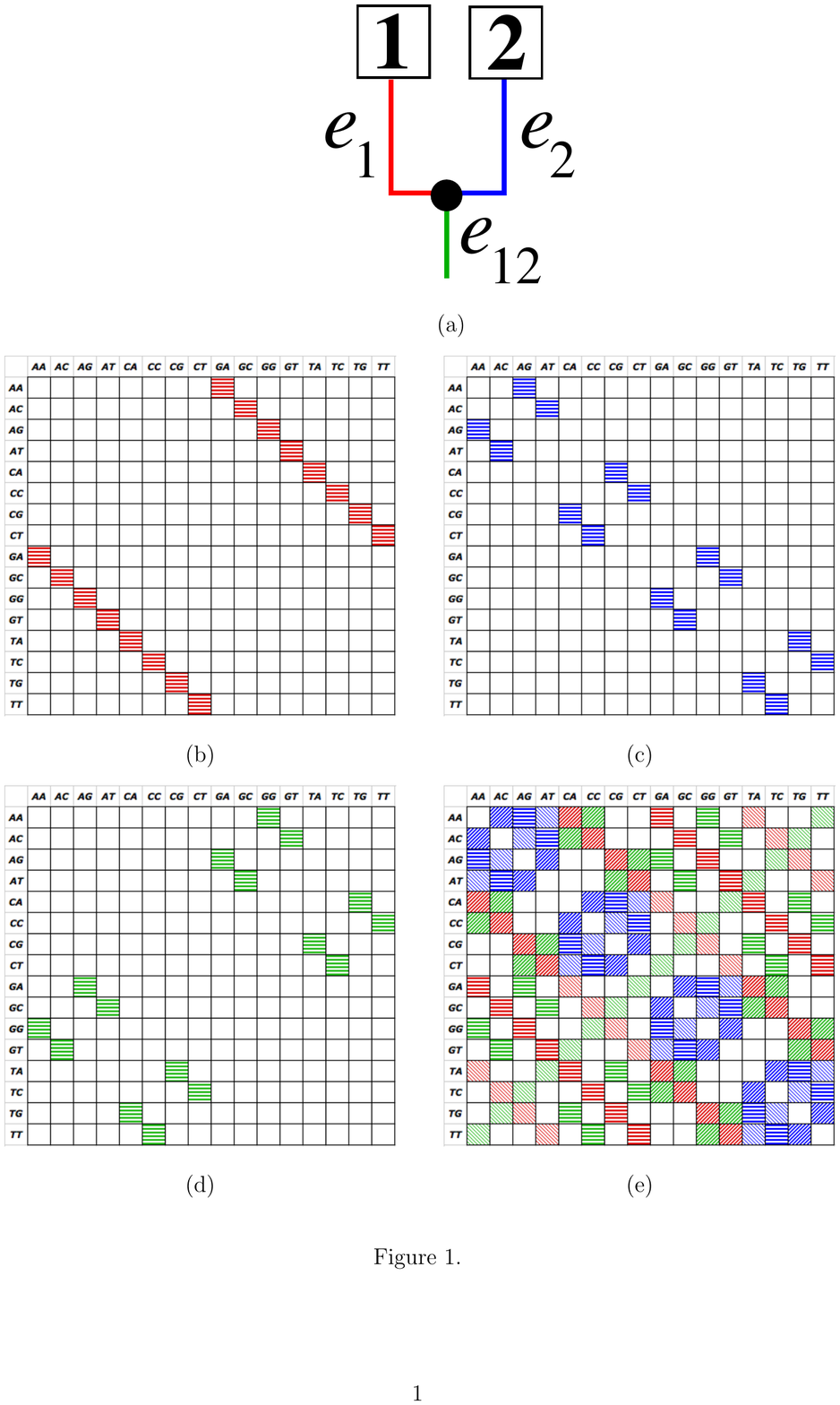}
  \caption{(a) A rooted tree with Taxa 1 and 2.
(b) A transition $s_1$ on the left branch $e_1$ (the red branch) changes a character into exactly one new character as depicted by the red horizontal stripe cells of the permutation matrix $\sigma^{e_1,s_1}$. The matrix has 16 rows and 16 columns representing the possible characters for the alignment of two nucleotide sequences. The permutation matrices generated by $s_1$ for the right branch $e_2$ (blue) and for the branch leading to the ``root" $e_{12}$ (green) are displayed in (c) and (d), respectively.
The convex sum of all the weighted (by the relative branch length and the probability of the substitution type) permutation matrices generated by all substitution types for all branches is the OSM matrix of the tree ($M_\T$) as shown in (e).
Horizontal stripe cells represent the probability of the Transition $s_1$; diagonal stripes the Transversion $s_2$; and thin reverse diagonal stripes the Transversion $s_3$.
The colors of these cells indicate the relative branch lengths and follow the colors of the branches as in (a).
Thus, these colors also depict the branch origin of the substitutions.}
   \label{fig:Figure1}
\end{figure}

\subsection*{Figure \figCayley - The Cayley graph for the two-taxon tree from \figcite{\figOSM}(a).}

\begin{figure}[!htb]
  \center
  \includegraphics[width=0.7\textwidth]{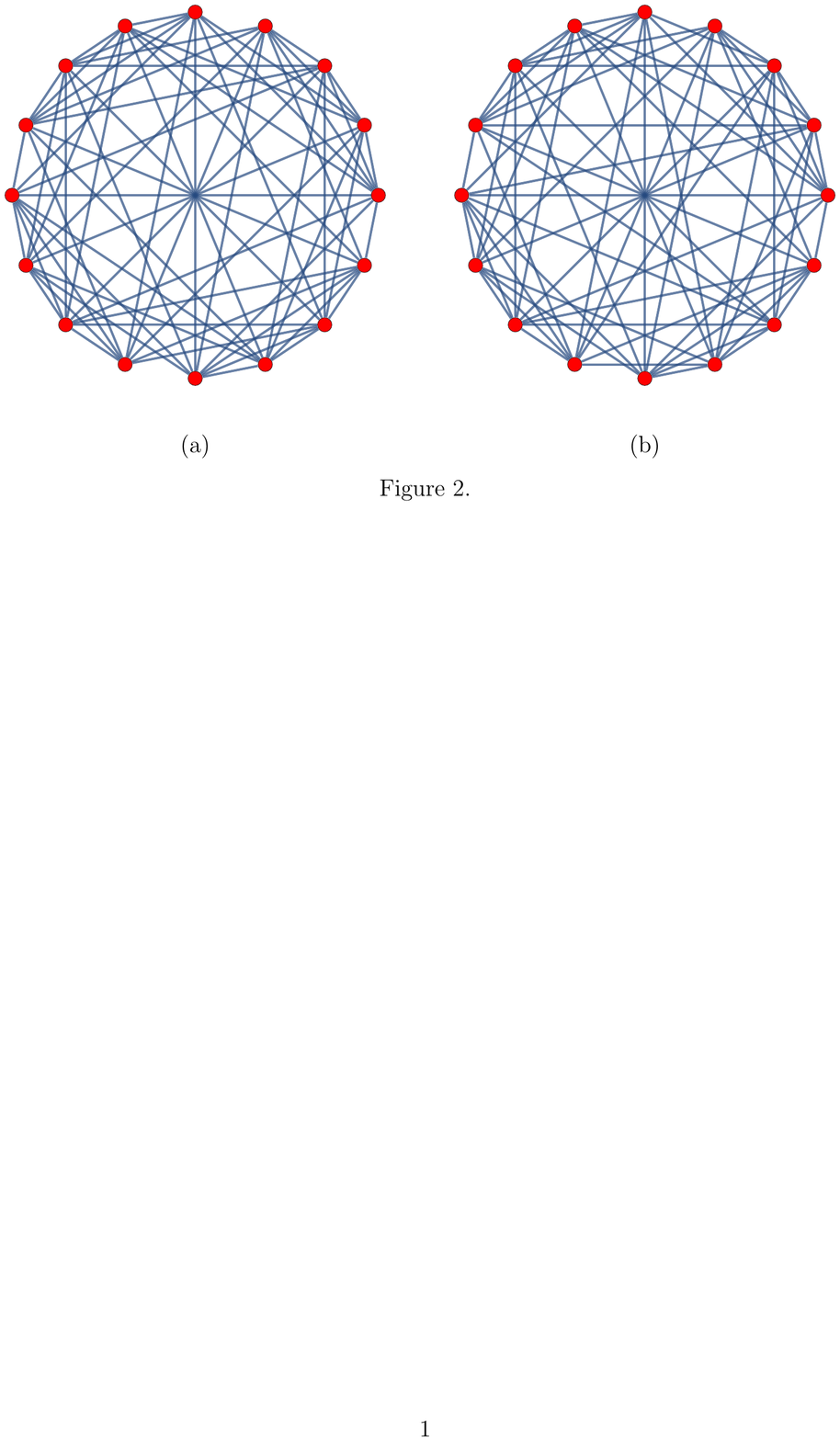}
  \caption{The vertices depict the characters in $\C_{DNA}^2$. Two vertices are connected by an edge if there is a permutation in $\mS^\T$ transforming one of the associated characters into the other. (a) depicts the Cayley graph for the Klein Four group $\Z_2\times\Z_2$, and (b) depicts the Cayley graph for the cyclic group $\Z_4$.}
   \label{fig:Figure2}
\end{figure}

\bigskip
\subsection*{Figure \figEgKlein - Computing the minimal number of substitutions to translate a character into another one.}

\begin{figure}[!htb]
  \center
  \includegraphics[width=0.7\textwidth]{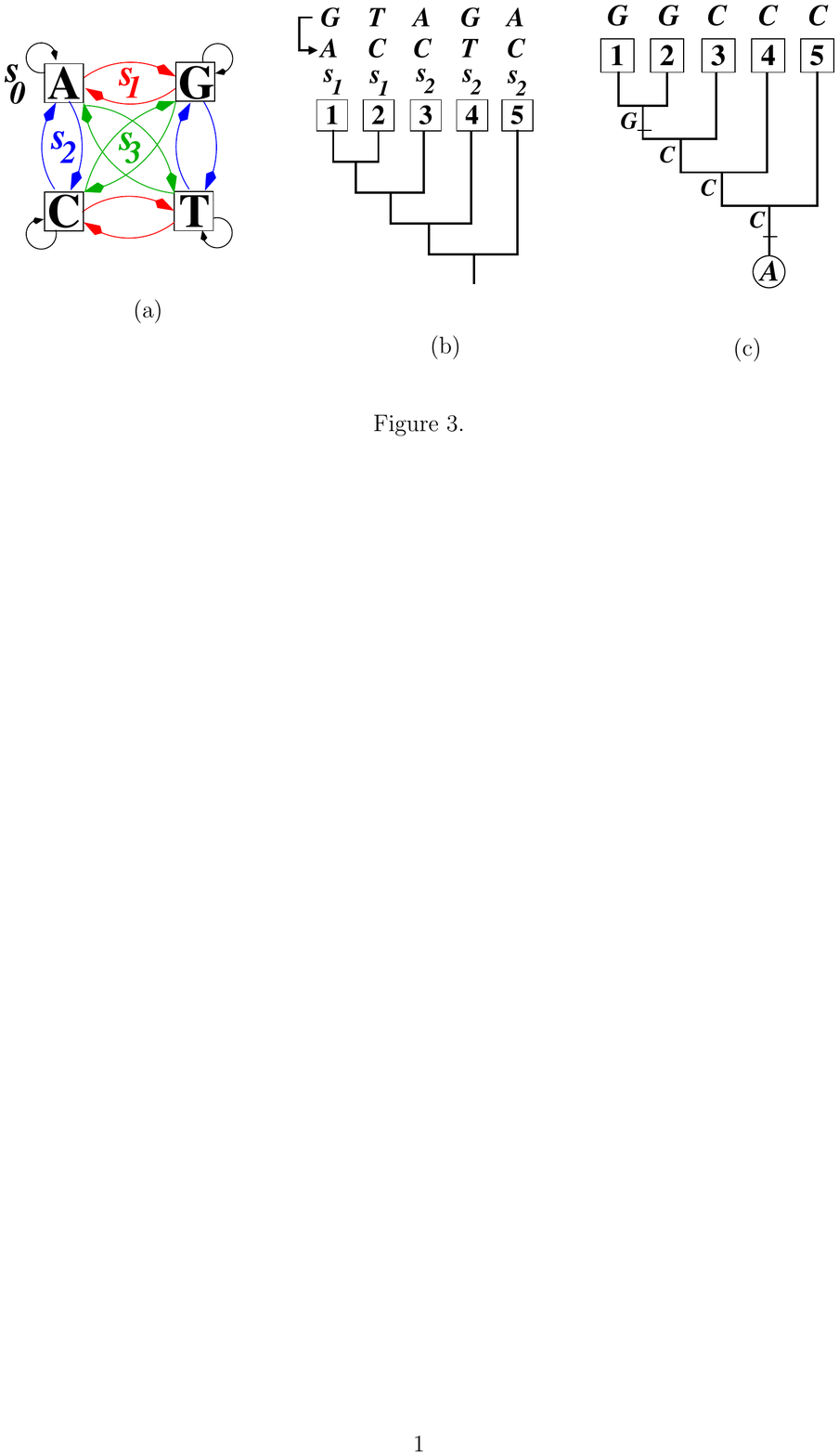}
  \caption{(a) depicts the Klein-four group $\mS_{K3ST}$, which consists of the identity $s_0$ and the three substitution types $s_1, s_2, s_3$ from the K3ST model. (b) In order to convert the character $GTAGA$ into $ACCTC$ under $\mS_{K3ST}$, we need to introduce the operation $\bm{s} :=(s_1, s_1, s_2, s_2, s_2)$. As the operation $\bm{s}$ also translates the constant character $AAAAA$ to $GGCCC$, converting $GTAGA$ into $ACCTC$ is equivalent to evolving the character state $A$ at the root along the tree to obtain the character $GGCCC$ at the leaves. The Fitch algorithm applied to the latter produces a unique most parsimonious solution of two substitutions as depicted by (c).}
   \label{fig:Figure3}
\end{figure}

\bigskip
\subsection*{Figure \figEgCyclic - Converting one character into another character using the cyclic group.}

\begin{figure}[!htb]
  \center
  \includegraphics[width=0.7\textwidth]{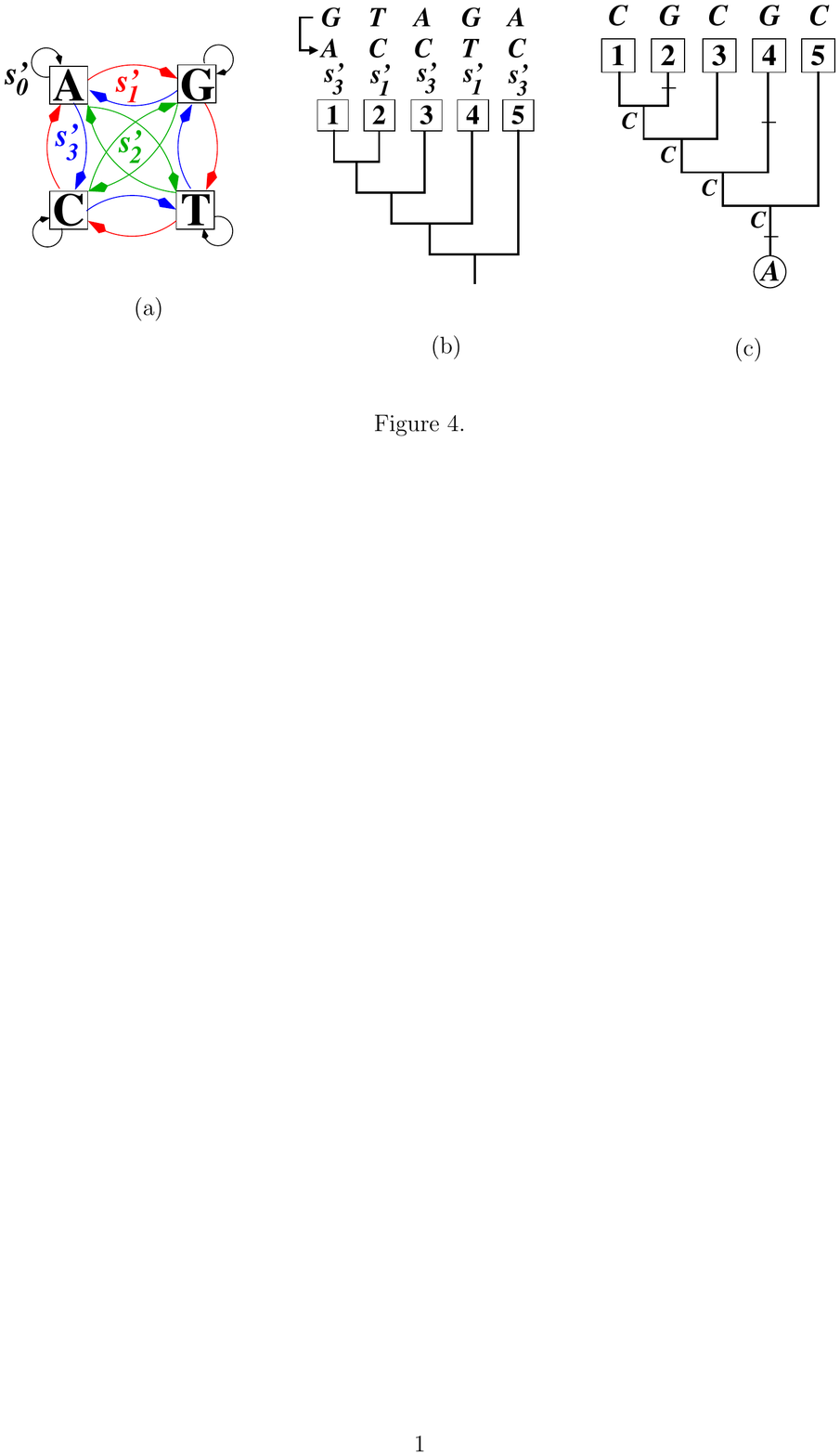}
  \caption{(a) depicts the cyclic group $\mS_c$, which consists of the identity $s'_0 \equiv s_0$ and the three substitution types $s'_1, s'_2, s'_3$ for nucleotide character states. (b) In order to convert the character $GTAGA$ into $ACCTC$ using this group, we need to introduce the operation $\bm{s}' :=(s'_3,s'_1,s'_3, s'_1, s'_3)$. As the operation $\bm{s}'$ also transforms the constant character $AAAAA$ to $CGCGC$, converting $GTAGA$ into $ACCTC$ is equivalent to evolving the character state $A$ at the root along the tree such that the character $CGCGC$ is attained at the leaves. The Fitch algorithm applied to the latter produces a unique most parsimonious solution of three substitutions as depicted by (c). }
   \label{fig:Figure4}
\end{figure}

\bigskip
\subsection*{Figure \figAminoAcid - Matrices illustrate the four Abelian groups for the twenty-state amino acid alphabet.}

\begin{figure}[!htb]
  \center
  \includegraphics[width=0.7\textwidth]{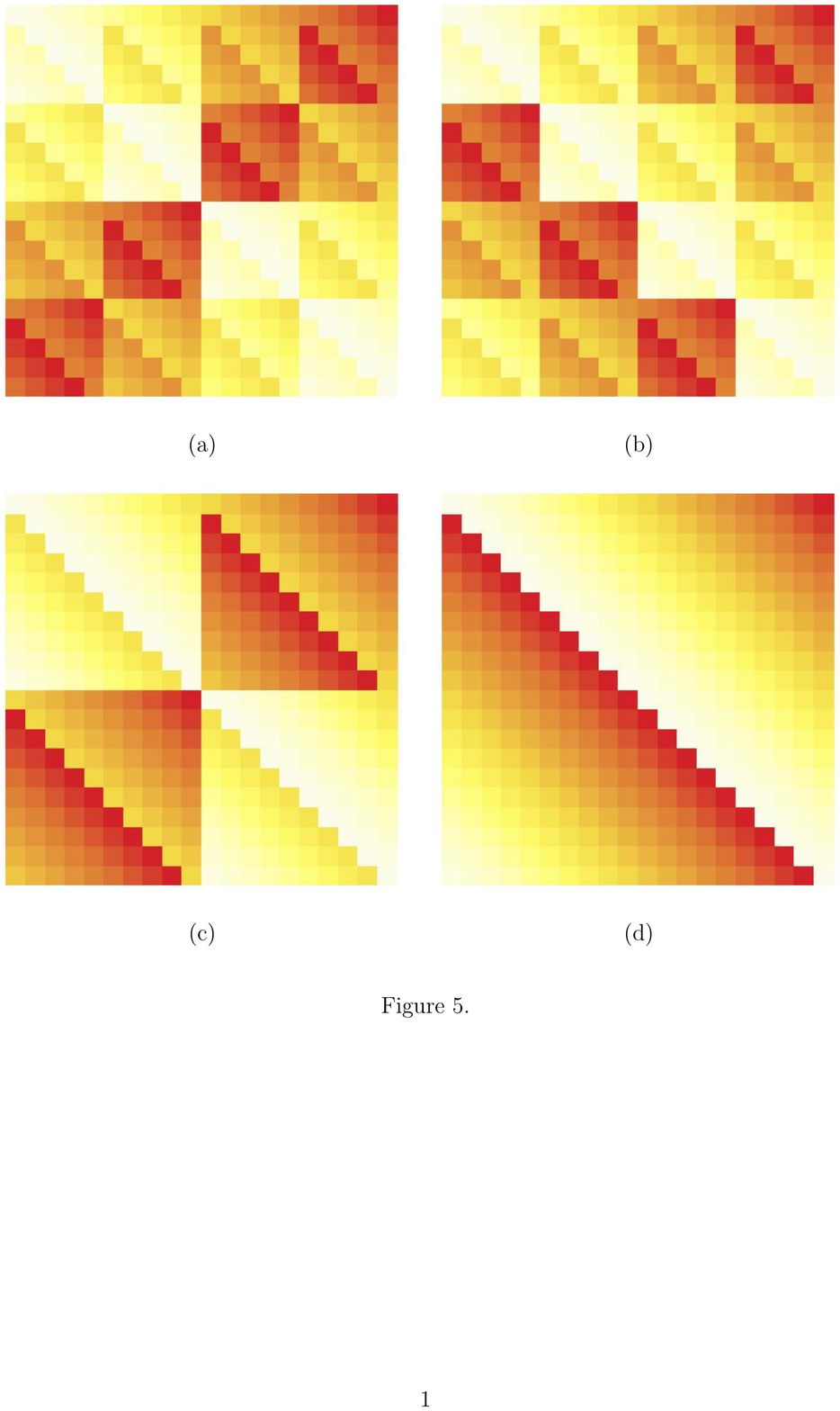}
  \caption{(a) the $\Z_2 \times \Z_2 \times \Z_5$ group, (b) $\Z_4 \times \Z_5$, (c) $\Z_2 \times \Z_{10}$, and (d) $\Z_{20}$.  In each matrix, the 20 different colors ranging from light yellow to dark red can be regarded to represent 20 substitution types, i.e. 20 operations including the identity acting on the amino acid character states or the corresponding probabilities of these substitution types. In the latter case, the matrices are all doubly stochastic.}
   \label{fig:Figure5}
\end{figure}




%
\end{bmcformat}
\end{document}